\newif\ifFull
\Fulltrue
\newif\ifSyntaxChecking

\ifFull
\documentclass[a4paper,11pt,twoside]{article}

\ifFull
\usepackage[margin=1in]{geometry}
\fi

\usepackage{fancyhdr}
\else
\documentclass[runningheads,orivec]{llncs}
\fi

\usepackage[labelfont=bf]{caption}

\usepackage[utf8]{inputenc}
\usepackage{subcaption}
\usepackage{wrapfig}

\usepackage{xcolor}

\usepackage{xspace}
\usepackage{enumerate}

\usepackage{microtype}

\usepackage{amsmath}
\ifFull
\usepackage{amsthm}
\fi
\usepackage{amssymb}
\usepackage{amstext}
\usepackage{esvect}
\usepackage{amsopn}

\usepackage{hyperref}
\usepackage{graphicx}
\usepackage[ruled,vlined]{algorithm2e}
\DontPrintSemicolon
\SetCommentSty{textrm}
\SetKwComment{tcp}{$\triangleright$ }{}
\SetKwFor{ForAll}{for all}{do}{}
\usepackage[basic]{complexity}

\usepackage{paralist}
\usepackage{enumitem}

\ifSyntaxChecking
\input prepareForSyntaxChecking

\else
\usepackage[left, pagewise, displaymath, mathlines]{lineno}
\newcommand*\patchAmsMathEnvironmentForLineno[1]{%
  \expandafter\let\csname old#1\expandafter\endcsname\csname #1\endcsname
  \expandafter\let\csname oldend#1\expandafter\endcsname\csname end#1\endcsname
  \renewenvironment{#1}%
     {\linenomath\csname old#1\endcsname}%
     {\csname oldend#1\endcsname\endlinenomath}}%
\newcommand*\patchBothAmsMathEnvironmentsForLineno[1]{%
  \patchAmsMathEnvironmentForLineno{#1}%
  \patchAmsMathEnvironmentForLineno{#1*}}%
\AtBeginDocument{%
\patchBothAmsMathEnvironmentsForLineno{equation}%
\patchBothAmsMathEnvironmentsForLineno{align}%
\patchBothAmsMathEnvironmentsForLineno{flalign}%
\patchBothAmsMathEnvironmentsForLineno{alignat}%
\patchBothAmsMathEnvironmentsForLineno{gather}%
\patchBothAmsMathEnvironmentsForLineno{multline}%
}
\fi

\usepackage[textsize=footnotesize]{todonotes}
\SetVlineSkip{0pt}  
\def\comment#1{}%
\def\withcomments{%
  \newcounter{mycommentcounter}%
   \def\comment##1{\refstepcounter{mycommentcounter}%
    \ifhmode%
     \unskip%
     {\dimen1=\baselineskip \divide\dimen1 by 2 %
       \raise\dimen1\llap{\tiny
	{-\themycommentcounter-}}}\fi%
     \marginpar[{\renewcommand{\baselinestretch}{0.8}%
       \hspace*{-2em}\begin{minipage}{1.5\marginparwidth}\footnotesize%
[\themycommentcounter]:%
\raggedright ##1\end{minipage}}]{\renewcommand{\baselinestretch}{0.8}%
       \begin{minipage}{1.5\marginparwidth}\footnotesize%
[\themycommentcounter]: \raggedright%
##1\end{minipage}}}%
  }

\newcommand{\ldt}{\mathrel{.\,.}}

\ifFull

\else

\fi


\ifFull
\else
\spnewtheorem*{prprty}{Property}{\bfseries}{\itshape}
\fi

\ifFull
\newtheorem{theorem}{Theorem}

\newtheorem{observation}{Observation}

\newtheorem{lemma}{Lemma}

\theoremstyle{remark}

\else
\let\doendproof\endproof
\renewcommand\endproof{~\hfill$\qed$\doendproof}
\fi

\title{%
  \hskip -1em
  Linear-Time Algorithms for Maximum-Weight Induced
  Matchings
\hskip -1em\null\\
  and Minimum Chain Covers in Convex Bipartite Graphs}
 
\ifFull
\author{
  Boris Klemz\footnote{Institut f\"ur Informatik, Freie Universit\"at Berlin, Germany}
  \and
  Günter Rote\footnotemark[1]
} 
\date{November 13, 2017}
\else 
\author{Boris Klemz\inst{1} \and Günter Rote\inst{1}}  
\institute{Institute of Computer Science, Freie Universit\"at Berlin, Germany
\fi

\begin{document}

\maketitle
\thispagestyle{empty}

\begin{abstract}
A bipartite graph $G=(U,V,E)$ is \emph{convex} if the vertices in $V$ can be linearly
ordered such that for each vertex $u\in U$, the neighbors of $u$ are
consecutive in the ordering of $V$. An \emph{induced matching} $H$ of~$G$ is a matching such that
no edge of $E$ connects endpoints of two different edges of $H$.

We show that in a convex bipartite graph with $n$ vertices and $m$ \emph{weighted} edges,
an induced matching of maximum total weight can be computed in $O(n+m)$ time.

An \emph{unweighted} convex bipartite graph has a representation of
size $O(n)$ that records for each vertex $u\in U$ the first and last neighbor
in the ordering of $V$.
Given such a \emph{compact representation}, we compute an induced matching of maximum cardinality in
$O(n)$ time.

In convex bipartite graphs, maximum-cardinality induced matchings are dual to minimum \emph{chain covers}.
A chain cover is a covering of the edge set by \emph{chain subgraphs}, that is, subgraphs that do not contain induced matchings of more than one edge.
Given a compact representation, we compute a representation of a minimum chain cover in $O(n)$ time.
If no compact representation is given, the cover can be computed in $O(n+m)$ time.

All of our algorithms achieve optimal running time for the respective
problem and model.
Previous algorithms considered only the unweighted case, and
the best algorithm for computing a maximum-cardinality induced matching or a minimum chain cover
in a convex bipartite graph
had a running time of $O(n^2)$.
\end{abstract} 

\section{Introduction}
\label{sec:intro}

\paragraph{Problem Statement.}

A bipartite graph $G=(U,V,E)$ is \emph{convex} if
$V$ can be numbered as
$\{1,2,\ldots,n_V\}$ so that the neighbors of
every vertex $i\in U$ form an \emph{interval}
$\{L^i,L^i+1,\ldots,R^i\}$, see Figure~\ref{fig:example-convex}(a).
For such graphs, we consider the problem of computing an induced
  matching (a)~of maximum cardinality or (b)~of maximum total weight,
for graphs with edge weights.

An \emph{induced matching} $H\subseteq E$ is a matching that results
as a subgraph induced by some subset of vertices.
This amounts to requiring
that
no edge of $E$ connects endpoints of two different edges of $H$, see Figure~\ref{fig:example-convex}(a).
In terms of the line graph, an induced matching is an independent set
in the \emph{square} of the line graph.
The square of a graph connects every pair of nodes whose distance
is one or two.
Accordingly, we call two edges of $E$ \emph{independent} if they can appear
together in an induced matching, or in other words, if their endpoints induce
a $2K_2$ (a~disjoint union of two edges) in $G$.
Otherwise, they are called \emph{dependent}.
  \begin{figure}[tb]
    \centering

\includegraphics[width=\columnwidth]{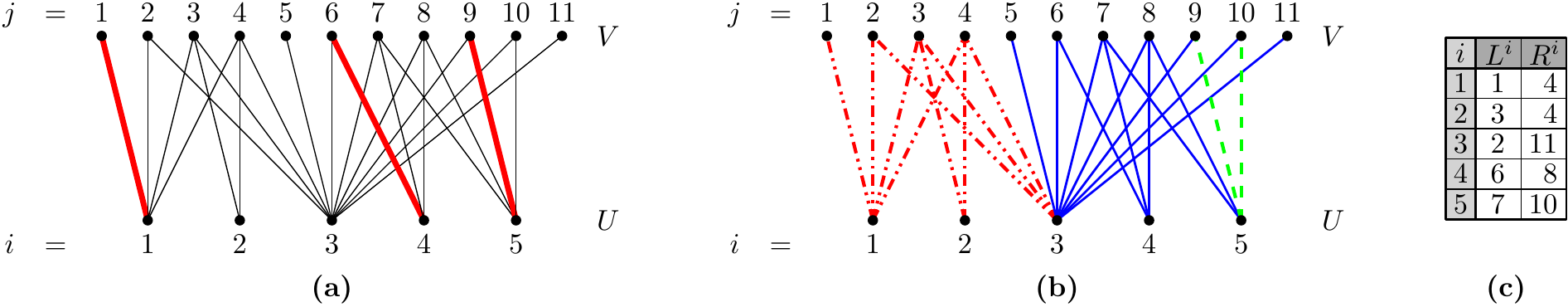}
   
\caption{\textbf{(a)} A convex bipartite graph~$G=(U,V,E)$ containing an
  induced matching $H$ of size~3.
Since we use successive natural numbers as elements of $U$ and $V$, we
will explicitly indicate whether we regard a number $x$ as a vertex of $U$ or
of $V$.
  There is no induced matching with more than 3 edges:
  vertex~$3\in U$ is adjacent to all vertices of~$V$ except~$1\in
  V$. Thus, if we match $3\in U$, this can only lead to induced matchings of size
  at most~$2$. Furthermore, we cannot simultaneously match $1\in U$ and
  $2\in U$ since every neighbor of
$2\in U$ is also adjacent to~$1\in U$.
  \textbf{(b)} A minimum chain cover of~$G$ with~$3$ chain
  subgraphs~$Z_1,Z_2,Z_3$ (in different colors and dash styles), providing an independent proof that $H$ is optimal. Here, $Z_1,Z_2,Z_3$ have disjoint edge
  sets, which is not necessarily the case in general. \textbf{(c)}~The compact representation of~$G$.
}
    \label{fig:example-convex}
  \end{figure}

In convex bipartite graphs, maximum-cardinality induced matchings are dual to minimum chain covers.
A \emph{chain graph} $Z$ is a bipartite graph
that contains no induced matching of more than one edge, i.\,e., it contains no pair of independent edges.
(Chain graphs are also called \emph{difference graphs}~\cite{DBLP:journals/dam/HammerPS90} or \emph{non-separable} bipartite graphs~\cite{DBLP:journals/jgt/Ding91}.)
A \emph{chain cover} of a graph~$G$ with edge set $E$ is a set of
chain subgraphs~$Z_1,Z_2,\ldots,Z_W$ of~$G$ such that the union of the
edge sets of~$Z_1,Z_2,\ldots,Z_W$ is~$E$, see
Figure~\ref{fig:example-convex}(b).
A chain cover with $W$ chain subgraphs provides an obvious {certificate}
that the graph
cannot
contain an induced matching with more than $W$ edges. We will elaborate on this
aspect
of a chain cover as a certificate of optimality in Section~\ref {certificate}.
A \emph{minimum} chain cover of~$G$ is a chain cover with a smallest possible number of chain subgraphs.
In a convex bipartite graph~$G$, the maximum size of an induced
matching is equal to the minimum number of chain subgraphs of a chain cover~\cite{DBLP:journals/tcs/YuCM98}.

We denote the number of vertices by
$n_U=|U|$,
$n_V=|V|$,
$n=n_U+n_V$,
and the number of edges by
$m=|E|$.
If a convex graph is given as an ordinary bipartite graph without the proper
numbering of $V$, it can be transformed into this form in linear time
$O(n+m)$~\cite{DBLP:journals/jcss/BoothL76}. (In terms of the
bipartite adjacency matrix, convexity is the well-known consecutive-ones property.)
Unweighted convex bipartite graphs have a natural implicit
representation~\cite{spinrad2003efficient} of size $O(n)$, which is
often called a \emph{compact representation}~\cite{DBLP:journals/mst/Hung12,DBLP:journals/algorithmica/SoaresS09}: every interval $\{L^i,L^i+1,\ldots,R^i\}$ is given by its endpoints $L^i$ and $R^i$, see Figure~\ref{fig:example-convex}(c).
Since the numbering of~$V$ can be computed in~$O(n+m)$
time, it is easy to obtain a compact representation in total time~$O(n+m)$~\cite{DBLP:journals/algorithmica/SoaresS09,STEINER199691}.
The chain covers that we construct will consist of
convex bipartite subgraphs with the same ordering of $V$ as the
original graph.
Thus, we will be able to use the same representation for the chain
graphs of a chain cover.

\paragraph{Related Work and Motivation.}

The problem of finding an induced matching of maximum size was first considered by Stockmeyer and Vazirani~\cite{DBLP:journals/ipl/StockmeyerV82} as the ``risk-free marriage problem'' with applications in interference-free network communication.
The decision version of the problem is known to be $\NP$-complete in many restricted graph classes~\cite{DBLP:journals/dam/Cameron89,DBLP:journals/ipl/Lozin02,DBLP:journals/algorithmica/KoblerR03}, in particular bipartite graphs~\cite{DBLP:journals/dam/Cameron89,DBLP:journals/ipl/Lozin02} that are $C_4$-free~\cite{DBLP:journals/ipl/Lozin02} or have maximum degree~$3$~\cite{DBLP:journals/ipl/Lozin02}.
On the other hand, it can be solved in polynomial time in chordal graphs~\cite{DBLP:journals/dam/Cameron89}, weakly chordal graphs~\cite{DBLP:journals/dm/CameronST03}, trapezoid graphs, $k$-interval-dimension graphs and co-comparability graphs~\cite{DBLP:journals/dam/GolumbicL00}, amongst others.
For a more exhaustive survey we refer to~\cite{DBLP:journals/jda/DuckworthMZ05}.

The class of convex bipartite graphs was introduced by Fred
Glover~\cite{glover1967maximum}, who motivates the computation of
matchings in these graphs with industrial manufacturing
applications. Items that can be matched when some quantity fits up to
a certain tolerance naturally lead to convex bipartite graphs.
The computation of matchings in convex bipartite graphs also corresponds to a scheduling problem of tasks of discrete length on a single disjunctive resource~\cite{DBLP:journals/informs/Katriel08}.
The problem of finding a (classic, not induced) matching of maximum cardinality in convex bipartite graphs has been studied extensively~\cite{glover1967maximum,STEINER199691,GALLO198431} culminating in an $O(n)$ algorithm when a compact representation of the graph is given~\cite{STEINER199691}.
Several other combinatorial problems have been studied in convex
bipartite graphs. While some problems have been shown to be
$\NP$-complete even if restricted to this graph
class~\cite{ASDRE2007248}, many problems that are $\NP$-hard in
general can be solved efficiently in convex bipartite graphs. For example, a maximum independent set can be found in~$O(n)$ time (assuming a compact representation)~\cite{DBLP:journals/algorithmica/SoaresS09} and the existence of Hamiltonian cycles can be decided in~$O(n^2)$ time~\cite{DBLP:journals/dm/Muller96a}. For a comprehensive summary we refer to~\cite{DBLP:journals/mst/Hung12}.

One of the applications given by Stockmeyer and
Vazirani~\cite{DBLP:journals/ipl/StockmeyerV82} for the induced
matching problem can be stated as follows. We want to test (or use) a
maximum number of connections between receiver-sender pairs in a
network. However, testing a particular connection produces noise so
that no other node in reach may be tested simultaneously. We remark
that this type of motivation extends very naturally to convex
bipartite graphs when we consider \emph{wireless} networks in which
nodes broadcast or receive messages in specific frequency ranges.
Further, \emph{weighted} edges can model the importance of connections.

\paragraph{Previous Work.}

Yu, Chen and Ma~\cite{DBLP:journals/tcs/YuCM98} describe an
algorithm that finds both a maximum-cardinality induced matching and a minimum chain cover in a convex bipartite graph in runtime~$O(m^2)$.
Their procedure is improved by Brandstädt, Eschen and Sritharan~\cite{DBLP:journals/tcs/BrandstadtES07}, resulting in a runtime of~$O(n^2)$.
Chang~\cite{DBLP:journals/dam/Chang03} computes maximum-cardinality induced matchings and minimum chain covers in~$O(n+m)$ time in bipartite permutation graphs, which form a proper subclass of convex bipartite graphs.
Recently, Pandey, Panda, Dane and
Kashyap~\cite{DBLP:conf/caldam/PandeyPDK17} gave polynomial algorithms
for finding a maximum-cardinality induced matching in circular-convex
and triad-convex bipartite graphs. These graph classes generalize
convex bipartite graphs.

\paragraph{Our Contribution.}

We improve the previous best~$O(n^2)$ algorithm \cite{DBLP:journals/tcs/BrandstadtES07} for maximum-cardinality induced matching and minimum chain covers in convex bipartite graphs in several ways.
In Section~\ref{sec:weighted} we give an algorithm for finding maximum-weight induced matchings in convex bipartite graphs with~$O(n+m)$ runtime.
The weighted problem has not been considered before.
In Section~\ref{sec:unweighted} we specialize our algorithm to find induced matchings of maximum cardinality in~$O(n)$ runtime, given a compact representation of the graph.
In Section~\ref{sec:chain} we extend this approach to obtain in~$O(n)$ time a compact representation of a minimum chain cover.
If no compact representation is given, our approach is easily adapted to produce a minimum chain cover in~$O(n+m)$ time.

All of our algorithms achieve optimal running time for the respective
problem and model. Our results for finding a maximum-cardinality induced
matching also improve the running times of the algorithms of Pandey
et al.~\cite{DBLP:conf/caldam/PandeyPDK17} for the circular-convex and
triad-convex case, as
they use the convex case as a
building block.

\section{Maximum-Weight Induced Matchings}
\label{sec:weighted}

In this section, we compute a maximum-weight induced matching of a given edge-weighted convex bipartite graph~$G=(U,V,E)$ in time~$O(n+m)$.
We generally write indices $i\in U$ as superscripts and indices $j\in V$ as
subscripts.
We consider $E$ as a subset of $U\times V$.
We assume that $V=\{1,\ldots,n_U\}$ is numbered as described in Section~\ref{sec:intro}
and the interval $\{L^i,L^i+1,\ldots,R^i\}\subseteq V$ of each
vertex~$i\in U$ is given
by the pair ($L^i$,$R^i$) of the left and right endpoint.
Each edge~$(i,j)\in E
$ has a weight~$C^i_j$.

Our dynamic-programming approach considers the following subproblems:
For an edge $(i,j)\in E$, we define
$W^i_j$ as the cost of the maximum-weight induced matching
that uses the edge $(i,j)$ and contains only edges in
$U\times \{1,\ldots,j\}$.
The following dynamic-programming recursion computes~$W^i_j$:
\begin{equation}
  \label{eq:recursion}
  W^{i}_{j}
 = C^{i}_{j}
 + \max \{\,W^{i'}_{j'}\mid
R^{i'}<j,\ j'<L^{i}
\,\}\cup\{0\}
\end{equation}
The range over which the maximum is taken is illustrated in Figure~\ref{fig:recursion}.
\begin{figure}[tb]
  \centering
  \includegraphics{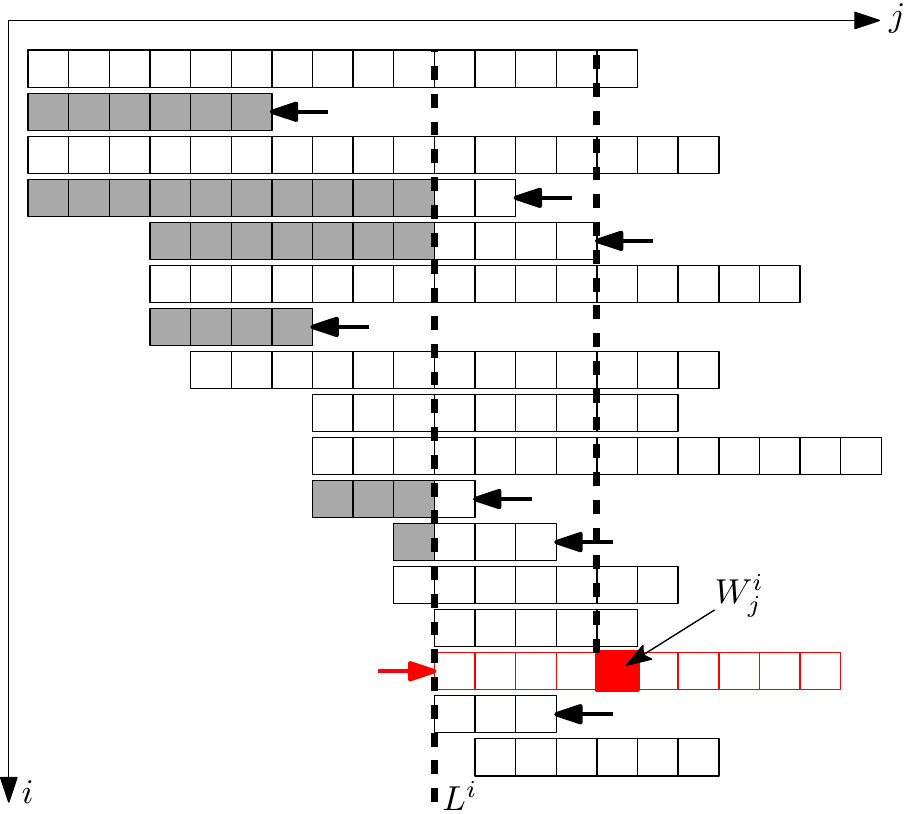}
  \caption{The table entries that go into the computation of $W^i_j$
 are shaded: They lie in rows that end to the left of
$W^i_j$ (marked by an arrow), and only the entries to the left of $L^i$ are considered.}
  \label{fig:recursion}
\end{figure}
In this recursion, we build the induced matching $H$ of weight
$W^{i}_{j}$ by adding the edge $(i,j)$ to some induced matching~$H'$ of weight
$W^{i'}_{j'}$.
We want $H$ to be an induced matching:
By construction, the edge $(i',j')$ is independent of
$(i,j)$, but we have to show that the other edges of $H'$ are also
independent of $(i,j)$.
In order to prove this
(Lemma~\ref{thm:RecursionCorrect}),
we use a transitivity relation
between independent edge pairs.

\newcommand{\theoremInducedChar}{
Two edges $(i,j)$ and $(i',j')$ are independent if and only if $j'\notin \lbrack L^i,R^i\rbrack$ and $j\notin \lbrack L^{i'},R^{i'}\rbrack$.
}
\begin{observation}
\label{thm:InducedChar}
\theoremInducedChar
\end{observation}

\newcommand{\theoremInducedTransitiv}{
  Let $(i'',j''),(i',j'),(i,j)\in E$ with $j''<j'<j$.
  Assume that $(i'',j'')$
and $(i',j')$ are independent, and $(i',j')$ and $(i,j)$ are
independent. Then $(i'',j'')$ and $(i,j)$ are independent.
}
\begin{lemma}
\label{thm:InducedTransitiv}
\theoremInducedTransitiv
\end{lemma}

\begin{proof}
By Observation~\ref{thm:InducedChar}, we have $j''\le R^{i''}<j'\le R^{i'}<j$ and $j''<L^{i'}\le j'<L^i\le j$. Thus, $j\notin \lbrack L^{i''},R^{i''}\rbrack$ and $j''\notin \lbrack L^i,R^i\rbrack$.
\end{proof}

\newcommand{\theoremRecursionCorrect}{
The recursion~\eqref{eq:recursion} is correct.
}
\begin{lemma}
\label{thm:RecursionCorrect}
\theoremRecursionCorrect
\end{lemma}

\begin{proof}
By Observation~\ref{thm:InducedChar}, any edge $(i',j')$ with $j'<j$
that is independent of~$(i,j)$ satisfies $R^{i'}<j$ and $j'<L^i$. By
Lemma~\ref{thm:InducedTransitiv}, all other edges~$(i'',j'')$ used to
obtain the matching value~$W^{i'}_{j'}$ are also independent of~$(i,j)$.
\end{proof}

We create a table in which we record the entries~$W^i_j$.
We assume that the intervals are sorted in nondecreasing order
by~$L^i$, that is,
$L^i\le L^h$ for $i<h$.
The values $W^i_{L^i},\ldots,W^i_{R^i}$ form the $i$-th \emph{row} of the table.
We fill the table row by row proceeding from~$i=1$ to~$i=n_U$.
Each row~$i$ is processed from left to right.

The only challenge in evaluating
\eqref{eq:recursion} is the
maximum-expression, for which
we introduce the notation
$M^{i}_{j}$.
\begin{equation*}
  M^{i}_{j} = \max \{\,W^{i'}_{j'}\mid
R^{i'}<j,\ j'<L^{i}
\,\}\cup\{0\}
\end{equation*}

We discuss the computation of the leftmost entry
$W^i_{L^i}$
later.
When we proceed from~$W^i_j$ to~$W^i_{j+1}$ we want to go incrementally from
$M^{i}_{j}$ to
$M^{i}_{j+1}$.
Direct comparison of the respective defining sets leads to
\begin{equation}\label{D-recursion}
  M^{i}_{j+1} = \max\ \{
M^{i}_{j}\}
\cup 
 \{\,W^{i'}_{j'}\mid
R^{i'}=j,\ j'<L^{i}
\,\}
\end{equation}
%

\begin{figure}[tb]
    \centering

\includegraphics{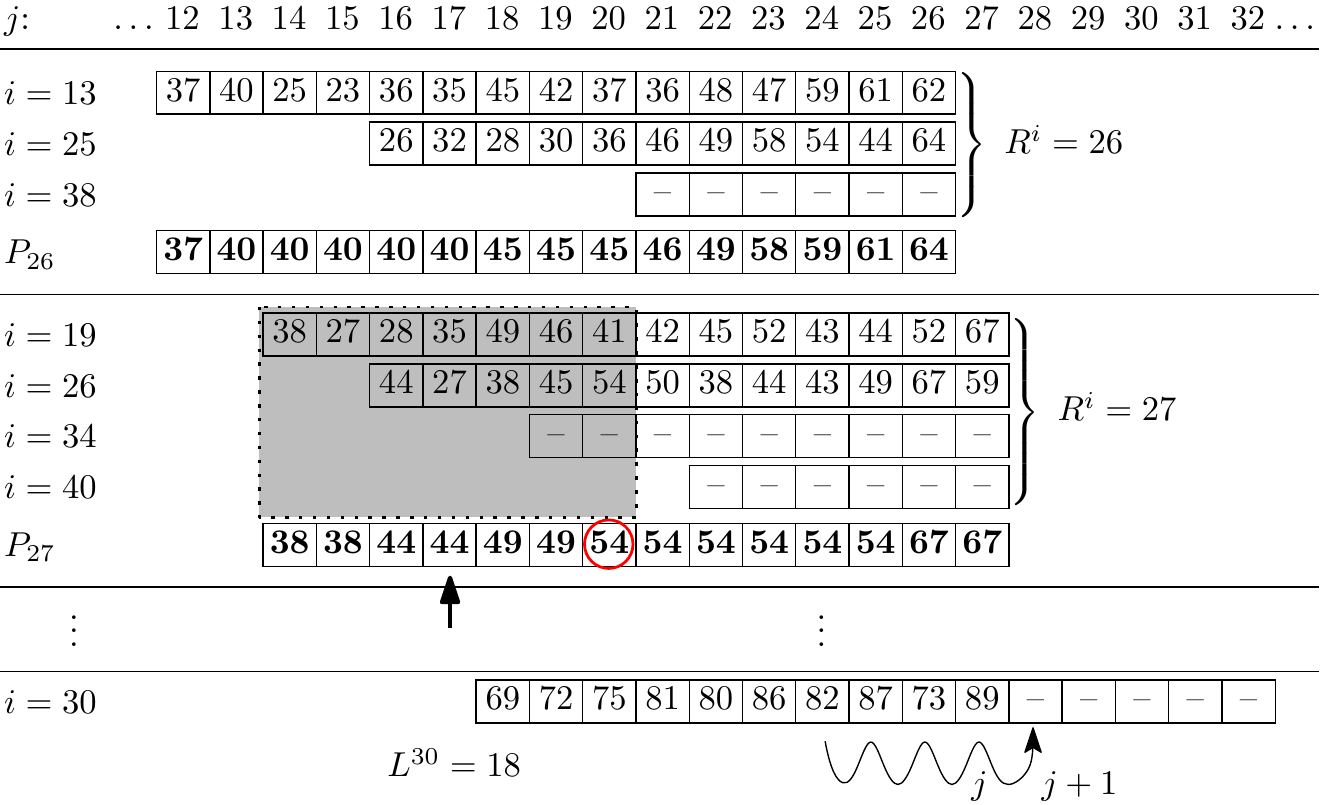}
   
    \caption{Example.
We are in the process of filling row $30$ from left to right. All rows with smaller
index $i$ have
been processed and are filled with the entries $W^i_j$.
Unprocessed entries are marked as ``--''.
 The figure does not show the rows in the order
 in which they are processed,
but intervals with the same right endpoint $R^i=r$ are grouped together.
The bold entries collect the provisional maxima $P_r$ in each group.
By way of example, the encircled entry $P_{27}[{20}]=54$
 is the maximum
among the shaded entries of the intervals that end at $R^i=27$,
ignoring the yet unprocessed entries.
As we proceed from $j=27$ to $j=28$ in row $30$, the intervals with
$R^i=27$ become relevant. The maximum usable entry from these
intervals is found in
position 17 of this array, because $17=L^{30}-1$.
The entry $P_{27}[{17}]=44$ is marked by an arrow.
The next entry
$W^{30}_{28}$ will be computed as
$C^{30}_{28}+\max\{P_{27}[17],P_{26}[17],\ldots,P_{17}[17]\}$.
(Some of these entries might not exist.)
We can observe that the minimum over which
$P_{27}[{17}]$ is defined involves no unprocessed entries
(Lemma~\ref{thm:RelevantEntries}).
When the next row $i=34$ in the group with $R^i=27$ is later
filled, it will be necessary to update $P_{27}$.
}
    \label{fig:example}
  \end{figure}

In order to evaluate the maximum of the second set in \eqref{D-recursion}
efficiently, we group intervals $i'$ with a common right endpoint
$R^{i'}=r$ together.
Let $S_r$ be the earliest startpoint of an interval with
endpoint~$r$.
If there are no intervals with endpoint~$r$, we set $S_r:=r$.
(It would be more logical to set $S_r:=r+1$ in this case, but
this choice makes the algorithm simpler.)
We maintain an array $P_r[j]$
 for
 $S_r\le j \le r$ that is defined as follows:
\begin{align*}
  P_r[j] := \max \{\, W^{i'}_{j'}\mid{}&R^{i'}=r,\\
&\text {row $i'$ has already been processed},\\
&j'\le j\,\} \cup \{0\}
\end{align*}
In a sense, $P_r[j]$ is a provisional version of the expression $\max \{\,W^{i'}_{j'}\mid
R^{i'}=r,\ j'<j
\,\}$,
which takes into account only the already processed rows. 
For~\eqref{D-recursion}, we need the entry
 $P_j[{L^i-1}]$, and we will see that
 all relevant entries have already been computed whenever we access
 this entry.
Thus, we rewrite \eqref{D-recursion}:
\begin{equation}
  \label{eq:D}
  M^{i}_{j+1} =
  \begin{cases}
 \max \{
M^{i}_{j}, P_j[{L^i-1}]\},
&\text{if $L^i-1\ge S_j$ and, thus, $P_j[{L^i-1}]$ is defined}
\\
M^{i}_{j},
&\text{otherwise}
  \end{cases}
\end{equation}
The condition $L^i-1\ge S_j$ ensures that
the array index $L^i-1$ does not exceed the left
boundary of the array $P_j$.
Also, the index $L^i-1$
 never exceeds the right
boundary $j$ of the array $P_j$, 
since $L^i< j+1\le R^i$, and therefore
$L^i-1\le j$.
 Thus, $P_j[{L^i-1}]$ is always defined when it is accessed.

\newcommand{\theoremRelevantEntries}{
When entry $W^i_{j+1}$ is processed,
 \eqref{D-recursion} and \eqref{eq:D} define the same quantity
$  M^{i}_{j+1}$.
}
\begin{lemma}
\label{thm:RelevantEntries}
\theoremRelevantEntries
\end{lemma}

\begin{proof}
We distinguish three cases.

\textbf{Case 1:}
No interval ends at~$j$, and accordingly, $S_j=j$.

In this case
 $M^i_{j+1}=M^i_j$ in \eqref{D-recursion} since its rightmost set is empty.
Since~$L^i< j+1\le R^i$, we have $L^i-1<S_j=j$ and, thus,
the right side of  \eqref{eq:D} evaluates also to $M^i_j$.

\textbf{Case 2:}
There exists an interval ending at~$j$, and~$L^i-1<S_j$.
The right side of  \eqref{eq:D} evaluates to $M^i_j$.
 In~\eqref{D-recursion},
intervals $i'$ that end at $R^{i'}=j$ have~$L^{i'}\ge
S_j>L^i-1$. Thus, an edge~$(i',j')$ with~$j'<L^i$ and~$R^{i'}=j$ does
not exist, and
the second set
 in~\eqref{D-recursion}
 is empty. Therefore,
 \eqref{D-recursion} evaluates to
 $M^i_{j+1}=M^i_j$.

\textbf{Case 3:}
There exists an interval ending at~$j$, and~$L^i-1\ge S_j$.
In this case, $P_j[{L^i-1}]$ is defined:
\begin{equation}\label{P-redundant}
P_j[{L^i-1}]=
\max \{\,W^{i'}_{j'}\mid R^{i'}=j,\ j'\le
L^i-1,\text{ row $i'$ already processed}\,\}
\end{equation}
For each entry 
$W^{i'}_{j'}$ with $j'<L^i$, we conclude that~$L^{i'}\le j'<L^i$ and,
thus, row $i'$ has already been processed. This means that the
condition that row $i'$ was processed is redundant, and
\eqref{P-redundant} coincides with the right side of~\eqref{D-recursion}.
\end{proof}

After processing row~$i$
with startpoint~$\ell=L^i$ and endpoint~$r=R^i$,
we have to update the values in $P_{r}[j]$.
This is straightforward. Figure~\ref{fig:example} illustrates the role of the arrays $P_r\lbrack j\rbrack$ when processing a row.

It remains to discuss the computation of the first value~$W^i_\ell$ of the row.
An edge~$(i',j'),j'<\ell$ and edge $(i,\ell)$ are independent if and only if the interval~$i'$ ends before~$\ell$, that is~$R^{i'}<\ell$.
Since we process the intervals in nondecreasing order by their
startpoints, it suffices to maintain a value~$F$ with the
maximum~$W^{i'}_{j'}$ 
in all \emph{finished} intervals: those intervals
$i'$ that end before~$\ell$. In other
words~$F=\max\{P_1[1],P_2[2],\ldots,P_{\ell-1}[\ell-1]\}$. This value
is easily maintained by updating~$F$ as~$\ell$ increases. The full
details
are stated as Algorithm~\ref{alg:weighted}.

The update of the array $P_r[j]$ in the second loop can be
integrated with the computation of $W^i_j$ in the first loop.
When this is done, the values $W^i_j$ need not be stored at all because they are
not used.
As stated earlier, when no interval ends at a point~$r\in V$, we set~$S_r=r$.
The array~$P_r$ consists of a single dummy entry~$P_r[r]=0$.
This way we avoid having to treat this special cases during the algorithm. 

\begin{algorithm}[tb]
\DontPrintSemicolon
\tcp{Preprocessing:}
\For{$r:=1$ \KwTo $n_V$}{
Find startpoint $S_r$ of the longest interval $[S_r,r]$ with endpoint $r$\;
Create an array  $P_r[S_r\ldt r]$ and initialize it to 0.\;
(If there is no such interval with endpoint $r$, set $S_r:=r$ and create an array with a single dummy entry $P_r[r]$ that will
remain at~0.)\;
}
\tcp{Main program:}
$F := 0$ \tcp*[h]{maximum entry in finished intervals}\;
\For{$\ell:=1$ \KwTo $n_U$}{
\tcp{$F = \max\{P_1[1],P_2[2],\ldots,P_{\ell-1}[\ell-1]\}$}
\ForAll
(\tcp*[h]{Process each interval $i$ that starts at $\ell$})
{rows $i\in U$ with $L^i=\ell$}{
$r := R^i$\;
\tcp{Process the $i$-th interval $[L^i,R^i]=[\ell,r]$ and fill row $i$ of the table:}
$M := M^i_{\ell}:= F$ \tcp*[h]{$M$ will be the current value of $M^{i}_{j}$}\;
$W^i_\ell := C^i_\ell + M$ \tcp*[h]{leftmost entry}\;
\For(\tcp*[h]{compute successive entries}){$j := \ell+1$ \KwTo $r$}{
\If{$S_j\le \ell-1$}{
$M := \max \{M, P_{j-1}[{\ell-1}]\}$ \tcp*[h]{$M^i_j := \max \{M^i_{j-1}, P_{j-1}[{\ell-1}]\}$}\;
}
$W^i_j := C^i_j + M$\;
}
\tcp{Go through the computed entries again to update the array $P_r$:}
$q := 0$ \tcp*[h]{the row maximum so far}\;
\For{$j:=\ell$ \KwTo $r$}{
$q := \max\{q,W^i_j\}$ \tcp*[h]{$q= \max\{0,W^i_{\ell},W^i_{\ell+1},\ldots,W^i_j\}$}\;
$P_r[j] := \max\{P_r[j],q\}$\;
}
}
$F := \max\{F, P_\ell[\ell]\}$ \tcp*[h]{update $F$ as $\ell$ is incremented}\;
}
\Return{$F$} \tcp*[h]{the maximum weight of an induced matching}\;
\caption{Weighted Maximum Matching}
\label{alg:weighted}
\end{algorithm}

We have described the computation of the \emph{value} of the optimal
matching.
It is straightforward to augment the program so that the optimal
matching itself can be recovered by backtracking how the optimal value
was obtained,
but this would clutter the program.

\begin{theorem}
A maximum-weight induced matching of an edge-weighted convex bipartite graph can be computed in $O(n+m)$ time.
\end{theorem}

\section{Maximum-Cardinality Induced Matchings}
\label{sec:unweighted}

For the unweighted version of the problem, we assume a compact representation of a convex bipartite graph~$G=(U,V,E)$, that is, for each~$i\in U$ we are given the startpoint~$L^i$ and endpoint~$R^i$ of its interval $\{L^i,L^i+1,\ldots,R^i\}$.
This makes it possible to obtain a linear runtime of $O(n)$.

The recursion \eqref{eq:recursion} can be specialized to the
unweighted case by setting $ C^{i}_{j}\equiv 1$.  
\begin{equation}
  \label{eq:recursion-unweighted}
  W^{i}_{j}
 = 1
 + \max\, \{\,W^{i'}_{j'}\mid
R^{i'}<j,\ j'<L^{i}
\,\}\cup\{0\}
\end{equation}
This
recursion has already been stated in~\cite{DBLP:journals/tcs/YuCM98} and~\cite{DBLP:journals/tcs/BrandstadtES07}
in a slightly different formulation.
Yu, Chen and Ma~\cite{DBLP:journals/tcs/YuCM98} describe it as a
greedy-like procedure that ``colors'' the edges of a bipartite graph
 with the values~$W^i_j$. From this coloring, they obtain both a maximum-cardinality induced matching and a minimum chain cover. The original implementation given in~\cite{DBLP:journals/tcs/YuCM98} runs in time~$O(m^2)$.
Brandstädt, Eschen and Sritharan~\cite{DBLP:journals/tcs/BrandstadtES07} give an improved implementation of the coloring procedure with runtime~$O(n^2)$.
Our Algorithm~\ref{alg:weighted} from Section~\ref{sec:weighted} obtains the values~$W^i_j$ in total time~$O(n+m)$.

Given a compact representation, we can exploit some structural properties of the filled dynamic-programming table to further improve the runtime to~$O(n)$.
The following observations were first given in~\cite{DBLP:journals/tcs/YuCM98} and~\cite{DBLP:journals/tcs/BrandstadtES07}.

\begin{lemma}[{\cite[Lemma 5]{DBLP:journals/tcs/YuCM98}}]
\label{thm:XPlusOne}
The values $W^i_j$ are nondecreasing in each row.
\end{lemma}
\begin{proof}
This is obvious from
  \eqref{eq:recursion-unweighted}, since 
the set
  over which the maximum is taken increases with~$j$.
\end{proof}

\begin{lemma}[{\cite[Lemma 3.3, Lemma 3.4]{DBLP:journals/tcs/BrandstadtES07}}]
\label{thm:AtMostTwo}
Each row contains at most two consecutive values.
\end{lemma}
\begin{proof}
Let $W^i_j$ be the largest value in some row $i$.
Then, if we take a corresponding matching 
 of size $W^i_j$,
it is easy to see that we
can remove the last two edges and replace them by
an arbitrary edge $(i,k)$. This proves that $W^i_{k}\ge W^i_j-1$.

More formally, we can argue by the recursion  \eqref{eq:recursion-unweighted}:
Assume there are values~$W_k^i\le W_j^i-2$ in row~$i$. By Lemma~\ref{thm:XPlusOne} we can assume~$k<j$.
By~\eqref{eq:recursion-unweighted},~$W_j^i=1+W_{j'}^{i'}=2+W_{j''}^{i''}$ with~$R^{i''}<j'<L^i$ for some~$i''<i'<i$. Thus, $j''\le R^{i''}<j'<L^i\le k$ and by definition of~$W_k^i$ according to \eqref{eq:recursion-unweighted} we have $W_j^i-2=W_{j''}^{i''}<W_k^i\le W_j^i-2$, which is a contradiction.
\end{proof}

\begin{algorithm}[tb]
\DontPrintSemicolon
Set $Q_1 := Q_2 := \cdots := Q_{n_U} := 0$\;
$F := 0$\;
\For{$\ell :=1$ \KwTo $n_V$}{
\ForAll
(\tcp*[h]{Process each interval $i$ that starts at $\ell$})
{rows $i\in U$ with $L^i=\ell$}{
$w := 
F+1$ \tcp{leftmost entry}
$t_w := {}$leftmost endpoint $R^{i'}$ of a row $i'$ that contains
an entry $W^{i'}_j=w$ with $j<L^i\equiv \ell$\;
\eIf(\tcp*[h]{There are two values $w$ and $w+1$ in this row:}){$t_w<R^i$}
{
\tcp{$W^i_j=w\phantom{{}+1}$ \ for $j=L^i,\dots,t_w$ }
\tcp{$W^i_j=w+1$ \ for $j=t_w+1,\ldots,R^i$ }
$Q_{R^i} := \max \{Q_{R^i},w+1\}$ \tcp{The largest entry is $w+1$.}
}
(\tcp*[h]{The same entry $w$ is used for the whole row.})
{
$Q_{R^i} := \max \{Q_{R^i},w\}$ \tcp{The largest entry is $w$.}
}
}
$F := \max \{F, Q_\ell\}$\ \tcp{update $F$ as $\ell$ advances}
}
\Return{$F$}
\caption{Unweighted Maximum Matching, initial version}
\label{alg:unweighted-1}
\end{algorithm}

\begin{algorithm}[tb]
\DontPrintSemicolon
 \nlset{$\bigtriangleup$} Set $t_1 := t_2 := \cdots := t_{n_U} :=
 n_V+1$ \tcp{The value $n_V+1$ acts like $\infty$.}
Set $Q_1 := Q_2 := \cdots := Q_{n_U} := 0$\;
$F := 0$\;
\For{$\ell :=1$ \KwTo $n_V$}{
\ForAll
(\tcp*[h]{Process each interval $i$ that starts at $\ell$})
{rows $i\in U$ with $L^i=\ell$}{
$w := 
F+1$ \tcp{leftmost entry}
 \nlset{$\bigtriangleup$} \tcp{$t_w$ is no longer computed from scratch}
\eIf(\tcp*[h]{There are two values $w$ and $w+1$ in this row:}){$t_w<R^i$}
{
\tcp{$W^i_j=w\phantom{{}+1}$ \ for $j=L^i,\dots,t_w$, }
\tcp{$W^i_j=w+1$ \ for $j=t_w+1,\ldots,R^i$. }
$Q_{R^i} := \max \{Q_{R^i},w+1\}$ \tcp{The largest entry is $w+1$.}
}
(\tcp*[h]{The same entry $w$ is used for the whole row.})
{
$Q_{R^i} := \max \{Q_{R^i},w\}$ \tcp{The largest entry is $w$.}
}
}
$F := \max \{F, Q_\ell\}$ \tcp{update $F$ as $\ell$ is incremented}
\nlset{$\bigtriangleup$} \ForAll{entries $W^{i'}_\ell$ in column $\ell$}{
 \nlset{$\bigtriangleup$} $w := W^{i'}_\ell$\;
 \nlset{$\bigtriangleup$} $t_w := \min \{t_w, R^{i'}\}$;
 }
}
\Return{$F$}
\caption{Unweighted Maximum Matching, second version}
\label{alg:unweighted-2}
\end{algorithm}

 Specializing Algorithm~\ref{alg:weighted} to the unweighted case
leads to a solution with $O(m)$ running time.
Our $O(n)$-time algorithm will follow the general scheme of Algorithm~\ref{alg:weighted},
with the following modifications.
\begin{itemize}
\item 
 In view of Lemmas \ref{thm:XPlusOne} and~\ref{thm:AtMostTwo}, we will not fill each
  row individually, but we will just determine the
  leftmost value $w
$ and the position where
  the entries switch from $w$ to $w+1$ (if any).
\item  The computation of the leftmost entry is exactly as
  in Algorithm~\ref{alg:weighted}.
\item The position 
 where
the entries of row $i$ switch from $w$ to $w+1$ can be determined from~\eqref{eq:recursion-unweighted}:
If
there is a row $i'$ containing an entry $w$ left of $L^i$, 
then $W^i_j$ must be $w+1$ as soon as $j>R^{i'}$.
The algorithm determines the
 \emph{threshold position} $t_w$ as
 the smallest right endpoint $R^{i'}$ under these constraints.
Then the entries  $w+1$ in row $i$ start at $j=t_w+1$ if these entries
are still part of the row.
\item 
 We do not maintain the whole array $P_r$ for each $r$, but only its last entry
 $P_r[r]$;
this is sufficient for updating $F$ and thus for computing
  the leftmost entries in the rows. We call this value~$Q_r$.
\end{itemize}
This leads to Algorithm~\ref{alg:unweighted-1}.

We will improve Algorithm~\ref{alg:unweighted-1} by maintaining the values $t_w$ instead
of computing them from scratch. 
We use the fact that the smallest value $w$ in the row is known, and
hence we can associate $t_w$ with the value $w$ instead of the row
index~$i$, as is already apparent from our chosen notation.
We update $t_w$ whenever $\ell$ increases.
The details are shown in Algorithm~\ref{alg:unweighted-2}. The differences to Algorithm~\ref{alg:unweighted-1} are marked by $\bigtriangleup$.


This still does not achieve $O(n)$ running time.
The final improvement comes from realizing that it is sufficient to
update $t_w$ when $W^{i'}_l$ is the \emph{leftmost} entry
  $w$ in row~$i'$.  The time when such an update occurs can be predicted when
  a row is generated.
To this end, we maintain a list $\mathcal{T}_j$ for $j=1,\ldots,n_V$ that
records the updates that are due when $\ell$ becomes~$j$.
This final version is Algorithm~\ref{alg:unweighted-3}.

\begin{algorithm}[tb]
\DontPrintSemicolon
\nlset{$\bigtriangleup$} Initialize lists $\mathcal{T}_1,\ldots,\mathcal{T}_{n_V}$ to empty lists\;
Set $t_1 := t_2 := \cdots := t_{n_U} := n_V+1$\;
Set $Q_1 := Q_2 := \cdots := Q_{n_U} := 0$\;
$F := 0$\;
\For{$\ell :=1$ \KwTo $n_V$}{
\ForAll
(\tcp*[h]{Process each interval $i$ that starts at $\ell$})
{rows $i\in U$ with $L^i=\ell$}{
$w := 
F+1$ \tcp{leftmost entry}
\eIf(\tcp*[h]{There are two values $w$ and $w+1$ in this row:}){$t_w<R^i$}
{
\tcp{$W^i_j=w\phantom{{}+1}$ \ for $j=L^i,\dots,t_w$, }
\tcp{$W^i_j=w+1$ \ for $j=t_w+1,\ldots,R^i$. }
\nlset{$\bigtriangleup$} \noindent\hbox spread -7pt{add $(w+1,R^i)$ to the list
  $\mathcal{T}_{t_w+1}$ }%
\tcp
{don't forget to update $t_{w+1}$ when $\ell$ reaches $t_w+1$}
\nlset{$\bigtriangleup$} \noindent\rlap{add $(w,R^i)$ to the list $\mathcal{T}_\ell$}%
\phantom{\hbox spread -7pt{add $(w+1,R^i)$ to the list
    $\mathcal{T}_{t_w+1}$ }}%
 \tcp
{don't forget to update $t_{w}$ when $\ell$ advances}
$Q_{R^i} := \max \{Q_{R^i},w+1\}$ \;
}
(\tcp*[h]{The same entry $w$ is used for the whole row.})
{
\nlset{$\bigtriangleup$} add $(w,R^i)$ to the list $\mathcal{T}_\ell$\;
$Q_{R^i} := \max \{Q_{R^i},w\}$ \;
}
}
$F := \max \{F, Q_\ell\}$ \tcp{update $F$ as $\ell$ advances}
\nlset{$\bigtriangleup$} \lForAll{
 $(w,r)\in \mathcal{T}_\ell$}
{$t_w := \min \{t_w, r\}$ \tcp*[h]{perform the necessary updates}}
}
\Return{$F$}
\caption{Unweighted Maximum Matching, final version}
\label{alg:unweighted-3}
\end{algorithm}
The runtime of Algorithm~\ref{alg:unweighted-3} is $O(n_U+n_V)$: Processing each interval $i$ takes
constant time and adds at most two pairs to the lists
 $\mathcal{T}$. Thus, processing
the lists $\mathcal{T}$ for updating the $t_w$ array takes also only
$O(n_U)$ time. 

Some simplifications are possible:
The addition of $(w,R^i)$ to the list $\mathcal{T}_\ell$ in the case
of two values can actually be omitted, as it leads to no decrease in
$t_w$: $t_w$ is already $<R^i$.
The algorithm could be further streamlined by observing that at most
two consecutive values of $t_w$ 
need to be remembered at any time.

Again, it is easy to modify the algorithm to return a maximum induced
matching in addition to its size.

\begin{theorem}
Given a compact representation, a maximum-cardinality induced matching of a convex bipartite graph can be computed in $O(n)$ time.
\end{theorem}

\section{Minimum Chain Covers}
\label{sec:chain}

In convex bipartite graphs, the size of a maximum-cardinality induced matching equals the number of chain subgraphs of a minimum chain cover~\cite{DBLP:journals/tcs/YuCM98}. In this section we use this duality and extend our Algorithm~\ref{alg:unweighted-3}  to obtain a minimum chain cover of a convex bipartite graph~$G=(U,V,E)$.

Let~$W^*$ be the cardinality of a maximum induced matching of~$G$.
Accordingly, the values~$W^i_j$ cover the range $\lbrace 1,\ldots,W^*\rbrace$.
We create~$W^*$ chain subgraphs~$Z_1,\ldots,Z_{W^*}$ of~$G$.
The edges $(i,j)$ with~$W^i_j=w$ will be part of the chain
subgraph~$Z_w$.

As already observed in~\cite{DBLP:journals/tcs/YuCM98},
the edges with a fixed value of $W^i_j$
may contain independent edges and, thus, do not necessarily constitute a chain graph.
Accordingly, Yu, Chen, and Ma~\cite{DBLP:journals/tcs/YuCM98} describe a strategy to extend the edge set for each value of~$W^i_j=w$ to a chain graph~$Z_w$.
Their original implementation
runs in time~$O(m^2)$.
Brandstädt, Eschen, and Sritharan~\cite{DBLP:journals/tcs/BrandstadtES07} give an improved implementation with runtime~$O(n^2)$.
We implement their strategy in~$O(n)$ time, given a compact representation.
The correctness was already shown in~\cite{DBLP:journals/tcs/YuCM98}.
We give a new independent proof. The following characterization is often used as an alternative definition of chain graphs:

\begin{lemma}\label{lem:chain-characterization}
  A bipartite graph $(\bar U,\bar V,\bar E)$ is a chain graph if and only if the sets of
 neighbors $\bar V(i)
 := \{\,j\in \bar V\mid (i,j)\in \bar E\,\}$ of the vertices
 $i\in \bar U$ form a chain in the inclusion order. \textup(Equal sets are
 allowed.\textup)
In other words, among any two sets
$\bar V(i)$ and $\bar V(i')$, one must be contained in the other.
\end{lemma}

\begin{proof}
  This is a direct consequence of
the fact that
edges $(i,j)$ and $(i',j')$ are independent if and only if $j'\notin
\bar V(i)$ and $j\notin \bar V(i')$.
\end{proof}
The condition that the neighborhoods must form a chain is apparently the
reason for calling these graphs \emph{chain graphs}, however, we did not find a reference for this.

We use~$U_w$ to denote the set of rows that contain entries $W^i_j=w$.
For every row~$i\in U_w$, we determine the beginning and ending 
points~$B^i_w,E^i_w$ with this color, that is, $W^i_j=w \iff B^i_w\le j\le E^i_w$.
We extend every such interval~$[B^i_w,E^i_w]$ to the left by choosing
a new starting point~$\hat B^i_w$ according to the formula
\begin{align}
  \label{eq:chainGoal}
  \hat B^i_w
&:=\min \lbrace B^i_w\rbrace \cup \lbrace\, B^{i'}_w\mid
i'\in U_w,\
E^{i'}_w<E^i_w\,\rbrace\\
\label{eq:chainGoal-hat}
&\mathrel{\phantom{:=}\llap{=}
}\min \lbrace B^i_w\rbrace \cup \lbrace\, \hat B^{i'}_w\mid
i'\in U_w,\
 E^{i'}_w<E^i_w\,\rbrace
\end{align}
The second expression uses the new values $\hat B$ on the
right-hand side.
It is easy to see that the two expressions are equivalent:
%
%
Using \eqref{eq:chainGoal} for the definition of
$\hat B^{i'}_w$, the expression 
\eqref{eq:chainGoal-hat} becomes
\begin{equation}
\label{eq:equalToHat}
\min \lbrace B^i_w\rbrace \cup \lbrace\, B^{i'}_w\mid
i'\in U_w,\
E^{i'}_w<E^i_w\,\rbrace \cup \lbrace B^{i''}_w\mid i''\in U_w,E^{i''}_w<E^{i'}_w<E^i_w,i'\in U_w\rbrace.
\end{equation}
The third set 
is contained in the second set, and thus, \eqref{eq:equalToHat} is
equal to~$\hat B^i_w$
according to~\eqref{eq:chainGoal}.

We construct the chain graph $Z_w$ as the graph with the extended intervals
$[\hat B^i_w,E^i_w]$.
Figure~\ref{fig:chain} shows an example.
\begin{figure}[tb]
  \centering
  \includegraphics{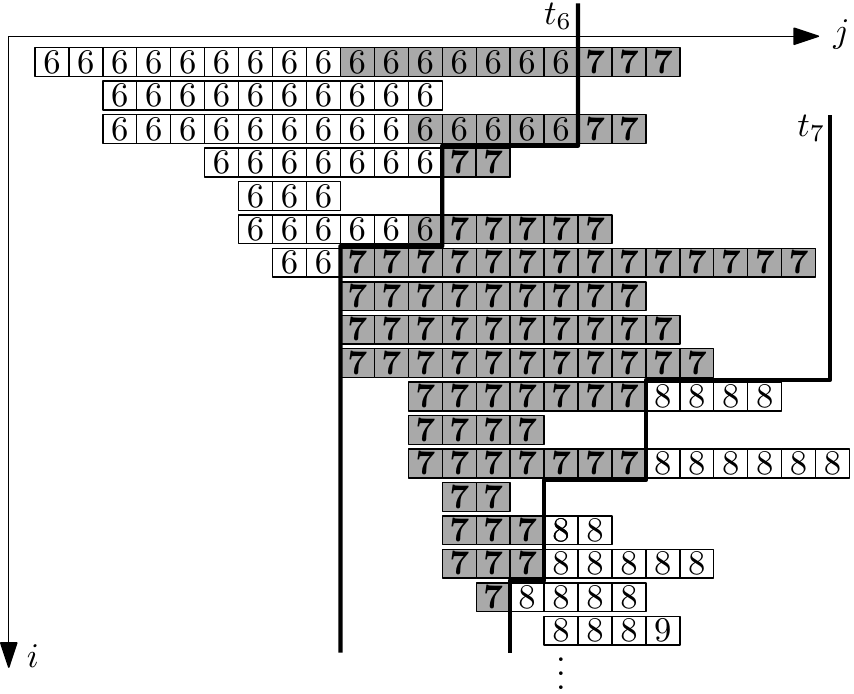}
  \caption{An example showing a section of the computation of $W^i_j$
by
Algorithm~\ref{alg:unweighted-3}.
The threshold values $t_6$ and $t_7$ are shown as they change
with the rows that are successively considered.
The shaded entries form the chain subgraph $Z_7$ that is used for the
chain cover.}
  \label{fig:chain}
\end{figure}
It is obvious by construction that these intervals satisfy the
conditions of a chain graph: By
Lemma~\ref{lem:chain-characterization}, we have to show that there are
no two intervals 
$[\hat B^i_w,E^i_w]$,
$[\hat B^{i'}_w,E^{i'}_w]$
with $\hat B^{i'}_w<\hat B^{i}_w$ and
$E^{i'}_w<E^{i}_w$. But if the last condition holds,
\eqref{eq:chainGoal-hat} ensures that
$\hat B^{i}_w\le \hat B^{i'}_w$.

The only thing that could go wrong is that
$\hat B^{i}_w$ becomes too small so that the chain graph is not
a subgraph of~$G$.
The following lemma shows that this is not the case.

\begin{lemma}
$\hat B^i_w\ge L^i$ for every~$i\in U_w$.
\end{lemma}

\begin{proof}
For the sake of contradiction, assume 
$\hat B^i_w< L^i$.
By~\eqref{eq:chainGoal}, there is a row~$i'\in U_w$ such that~$B^{i'}_w<L^i$ and $E^{i'}_w<E^i_w$.
Setting $j=E^i_w$
and $j'=B^{i'}_w$
 in the recursion~\eqref{eq:recursion-unweighted},
we conclude that $E^i_w\le R^{i'}$,
because otherwise,
\eqref{eq:recursion-unweighted} would imply $w=W^i_{E^i_w}\ge 1+
W^{i'}_{B^{i'}_w}= 1+w$.
Thus, $(i',E^i_w)$ is an edge of~$G$.
By Lemma~\ref{thm:AtMostTwo}, $W^{i'}_{E_w^i}=w+1$.
By~\eqref{eq:recursion-unweighted}, there is an edge~$(i'',j'')$ with
$W^{i''}_{j''}=w$, $R^{i''}<E_w^i$ and $j''<L^{i'}<L^i$.
Again by~\eqref{eq:recursion-unweighted}, such an edge $(i'',j'')$
would imply that $W^i_{E_w^i}\ge w+1$, a contradiction.
\end{proof}

Algorithm~\ref{alg:chain} carries out the computation
of~\eqref{eq:chainGoal}.
It 
processes the 
 triplets
$(B^i_w,E^i_w,w)$ in increasing order of the endpoints $E^i_w=r$. This can be
done in linear time, by first sorting the
 $O(n_U)$ triples
$(B^i_w,E^i_w,w)$
 into 
 $n_V$ buckets
 according to the value of $E^i_w$.
Thus, Algorithm~\ref{alg:chain} takes linear time~$O(n)$. By
Lemma~\ref{lem:chain-characterization}, the result is a chain cover,
which by duality is minimum.  Each row belongs to at most two chain
subgraphs, and thus the chain cover consists of at most $2n_U$ such
row intervals in total.  It is straightforward to extend
Algorithm~\ref{alg:unweighted-3} to compute the sets~$U_w$ and the quantities
$B^i_w,E^i_w$,
 and thus the cover can be constructed in $O(n)$
time in compressed form.

\begin{theorem}
Given a compact representation of a convex bipartite graph, a compact representation of a minimum chain cover can be computed in~$O(n)$ time.
\end{theorem}

Given a compact representation of a minimum chain cover, we can list all the edges of its chain subgraphs in~$O(n+m)$ time since every edge is contained in at most two chain subgraphs. As mentioned in the introduction, a compact representation of a convex bipartite graph can be computed in~$O(n+m)$ time~\cite{DBLP:journals/algorithmica/SoaresS09,STEINER199691,DBLP:journals/jcss/BoothL76}. Thus, Algorithm~\ref{alg:unweighted-3} and Algorithm~\ref{alg:chain} can also be used to obtain:

\begin{theorem}
A minimum chain cover of a convex bipartite graph can be computed in~$O(n+m)$ time.
\end{theorem}

\section{Certification of Optimality}
\label{certificate}
An induced matching $H$ together with a chain cover
of the same cardinality provides a \emph{certificate of optimality}, of size
$O(n)$.  As we will establish in the following discussion, it is easy
to \emph{check} this certificate for validity in linear time. This is easier than
\emph{constructing} the largest induced matching with our
algorithm. Thus, it is possible to establish correctness of the result
beyond doubt, for each particular instance of the problem, without
having to trust the correctness of our algorithms and their
implementations, see \cite{CertifyingAlgorithms} for a survey about
this concept.

\begin{algorithm}[tb]
\tcp{$U_w := \{\,i\in U \mid \text{row $i$ contains an entry $w$}\,\}$}

\tcp{Let $B^i_w$ and $E^i_w$ 
such that in row $i$, the entries with $W^i_j=w$ are those with $B^i_w \le j \le E^i_w$
}

Set $G_1:= G_2 := \cdots := G_{W^*} := n_V+1$ \tcp{The value $n_V+1$ acts like $\infty$}
  \For {$r := 1\ \KwTo\ n_v$}{
\tcp{We maintain the quantities $G_w\equiv\min \lbrace B^i_w\mid
  E^i_w<r\rbrace$ for $w=1,\ldots,W^*$.}
  \ForAll {$(B^i_w,E^i_w,w)$ 
with $E^i_w=r$}{
    $\hat B^i_w := \min\lbrace B^i_w,G_w\rbrace$
  }
  \ForAll(\tcp*[h]{update $G_w$ for the increment of $r$})
{$(B^i_w,E^i_w,w)$ 
 with $E^i_w=r$}{  
  $G_w:=\min\lbrace B^i_w,G_w\rbrace$
    }
  }
\caption{Constructing a chain graph 
$
\{\,(i,j)\mid i\in U_w,\, \hat B^i_w \le j \le E^i_w\,\}$,
$1\le w\le W^*$
}
\label{alg:chain}
\end{algorithm}

It is trivial to check whether the matching $H$ is
contained in the graph.
To test whether it forms an induced matching, we sort the edges
$(i,j)$ by $j$. This takes $O(n)$ time with bucket-sort.
Then, by Lemma~\ref{thm:InducedTransitiv}, it is sufficient to test
consecutive edges for independence, and each such test takes only
constant time according to
 Observation~\ref{thm:InducedChar}.

 To establish the validity of a chain cover
 $\{Z_1,\ldots,Z_{W^*}\}$,
 we need to check that the edges of~$G$ are covered
 and each
$Z_w$
is a chain subgraph.
The chain subgraphs
$Z_w=\{\,(i,j)\mid i\in U_w,\, \hat B^i_w \le j \le E^i_w\,\}$, for
$1\le w\le W^*$ are compactly represented by a set of at most $2n_U$
quadruples $(w,i,\hat B^i_w,E^i_w)$.  The following checking procedure works in
linear time for any chain cover as long as it consists of convex
bipartite subgraphs. It does not use any special properties of the
cover produced by our algorithm.

We sort the quadruples $(w,\hat B^i_w,-E^i_w,i)$
lexicographically. Then it is easy to check the chain graph property
using the characterization of Lemma~\ref{lem:chain-characterization}:
The intervals $[\hat B^i_w,E^i_w]$ that belong to a fixed chain graph
$Z_w$ (these are consecutive in the list) ought to be nested.
Since the starting points $\hat B^i_w$ are weakly increasing,
this
amounts to checking that the endpoints $E^i_w$ decrease weakly.

To check that the chain graphs are contained in $G$ and they collectively
cover~$G$, we sort the quadruples
$(i,\hat B^i_w,E^i_w,w)$.
The union of the intervals $[\hat B^i_w,E^i_w]$ that are the neighbors of
a fixed vertex $i\in U$
(these are consecutive in the list)
can be incrementally formed, and the resulting interval is compared against $[L^i,R^i]$. As soon
as a gap would form in this union, we can abort the test
, since the intervals are sorted by left endpoint and it
is then impossible to form a connected interval $[L^i,R^i]$.

The required lexicographic sorting operations can be carried out in $O(n)$ time
by bucket-sort.

\section{Outlook: Duality}

The existence of a pair of maximum induced
matchings and smallest chain covers with the same size
is a manifestation of strong duality between independents sets and
clique covers in perfect graphs.
We mentioned in the introduction that
our maximum induced matching problem 
is an instance of a maximum independent set problem in
the square of a line graph,
and the chain cover is a covering by cliques.
Yu, Chen and Ma~\cite{DBLP:journals/tcs/YuCM98} established
that the square of the line graph of a convex bipartite graph is a co-comparability graph.
Therefore, it is also a perfect graph.
It follows that the linear
program for maximizing the size of an induced matching is
totally dual integral.
As a corollary of this fact, we recover our strong duality result:
the existence of a \emph{primal} optimal solution (maximum induced
matching) and a \emph{dual} optimal solution (smallest chain cover) with
 matching objective function values.

This duality relation for perfect graphs extends to the weighted
version.
 Thus, there
should also be a weighted chain cover with the same weight as the
maximum weight of an induced matching.
It would be interesting to extend our
 primal
Algorithm~\ref{alg:weighted} in weighted graphs to a fast
combinatorial algorithm for finding
minimum-weight chain covers, as
Algorithm~\ref{alg:chain} does for the unweighted version.


{\small 
   \bibliographystyle{lncs-01-doi}
   \bibliography{abbrv,bib}
}
\end{document}